\documentclass[ssy,preprint]{imsart}

\usepackage[latin1]{inputenc}
\usepackage{amssymb,amsmath,amsthm,amsfonts,dsfont}
\usepackage{graphicx}
\usepackage{tikz}
\usetikzlibrary{patterns}
\usetikzlibrary{decorations.pathreplacing}
\usepackage{tikz-cd}
\usepackage{hyperref}
\usepackage[nodayofweek]{datetime}
\usepackage[normalem]{ulem}

\startlocaldefs
\newdateformat{mydate}{\shortmonthname[\THEMONTH]{ }\twodigit{\THEDAY}, \THEYEAR}

\numberwithin{equation}{section}

\newtheorem{theorem}{Theorem}[section]
\newtheorem{lemma}[theorem]{Lemma}

\newtheorem{proposition}[theorem]{Proposition}
\newtheorem{corollary}[theorem]{Corollary}
\theoremstyle{definition}

\newtheorem{remark}{Remark}[section]

\newtheorem{assumption}{Assumption}[section]

\endlocaldefs

\begin{document}

\begin{frontmatter}
\title{Stability, memory, and messaging tradeoffs in heterogeneous service systems}

\begin{aug}
\author[1]{\fnms{David} \snm{Gamarnik}\ead[label=e1]{gamarnik@mit.edu}}
\author[1]{\fnms{John N.} \snm{Tsitsiklis}\ead[label=e2]{jnt@mit.edu}}
\and
\author[2]{\fnms{Martin} \snm{Zubeldia}\ead[label=e3]{m.zubeldia.suarez@tue.nl}}

\address{David Gamarnik\\
John N. Tsitsiklis\\
Massachusetts Institute of Technology\\
\printead{e1}\\
\phantom{E-mail:\ }\printead*{e2}}

\address{Martin Zubeldia\\
Eindhoven University of Technology\\
\& University of Amsterdam\\
\printead{e3}}

\end{aug}

\vspace{3mm}
\today
\vspace{3mm}

\begin{abstract}
We consider a heterogeneous distributed service system, consisting of $n$ servers with unknown and possibly different  processing rates. Jobs with unit mean and independent processing times arrive as a renewal process of rate $\lambda n$, with $0<\lambda<1$, to the system. Incoming jobs are immediately dispatched to one of several queues associated with the $n$ servers. We assume that the dispatching decisions are made by a central dispatcher endowed with a finite memory, and with the ability to exchange messages with the servers.

We study the fundamental resource requirements (memory bits and message exchange rate) in order for a dispatching policy to be {\bf maximally stable}, i.e., stable whenever the processing rates are such that the arrival rate is less than the total available processing rate. First, for the case of Poisson arrivals and exponential service times, we present a policy that is maximally stable while using a positive (but arbitrarily small) message rate, and $\log_2(n)$ bits of memory. Second, we show that within a certain broad class of policies, a dispatching policy that exchanges $o\big(n^2\big)$ messages per unit of time, and with $o(\log(n))$ bits of memory, cannot be maximally stable. Thus, as long as the message rate is not too excessive, a logarithmic memory is necessary and sufficient for maximal stability.
\end{abstract}

%

\end{frontmatter}

\setcounter{tocdepth}{2}
\tableofcontents

\section{Introduction}
Distributed service systems are pervasive, from the checkout lines at the supermarket, to server farms for cloud computing. At a high level, many of these systems involve a stream of incoming jobs that are dispatched to a distinct queue associated with one of the servers (see Figure~\ref{fig:basicSetting} for a stylized model). Naturally, the behavior and performance of such systems depends on the dispatching policy.

 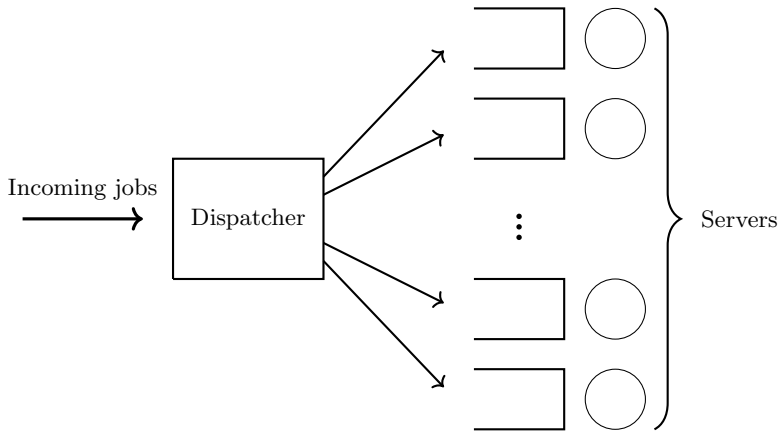
\begin{figure}[ht!]
 \centering
 \begin{tikzpicture}[scale=0.8]
    \draw [very thick,->] (-0.5,0) -> (1.5,0);
    \draw (0.5,0.5) node {Incoming jobs};
     \draw [thick] (2,-1) -- (2,1) -- (4.5,1) -- (4.5,-1) -- (2,-1);
     \draw (3.25,0) node {Dispatcher};
     \draw [thick, ->] (4.5,0.7) -- (6.5,2.8);
     \draw [thick, ->] (4.5,-0.7) -- (6.5,-2.8);
     \draw [thick, ->] (4.5,0.4) -- (6.5,1.4);
     \draw [thick, ->] (4.5,-0.4) -- (6.5,-1.4);
     \begin{scope}[yshift=3cm,xshift=5cm]
          \draw [thick] (2,0.5) -- (3.5,0.5) -- (3.5,-0.5) -- (2,-0.5);
          \draw (4.35,0) circle (5mm);
     \end{scope}
     \begin{scope}[yshift=1.5cm,xshift=5cm]
          \draw [thick] (2,0.5) -- (3.5,0.5) -- (3.5,-0.5) -- (2,-0.5);
          \draw (4.35,0) circle (5mm);
     \end{scope}
     \begin{scope}[xshift=5cm]
          \draw (2.75,0) node {\Huge\vdots};
     \end{scope}
     \begin{scope}[yshift=-3cm,xshift=5cm]
          \draw [thick] (2,0.5) -- (3.5,0.5) -- (3.5,-0.5) -- (2,-0.5);
          \draw (4.35,0) circle (5mm);
     \end{scope}
     \begin{scope}[yshift=-1.5cm,xshift=5cm]
          \draw [thick] (2,0.5) -- (3.5,0.5) -- (3.5,-0.5) -- (2,-0.5);
          \draw (4.35,0) circle (5mm);
     \end{scope}
          \draw [thick,decorate,decoration={brace,amplitude=10pt}]   (10,3.5) -- (10,-3.5) node [midway,right,xshift=.5cm] {Servers};
  \end{tikzpicture}
  \caption{Parallel server queueing system with a central dispatcher.}
    \label{fig:basicSetting}
  \end{figure}

While delay performance and stability are important factors when choosing how to operate these systems, the huge number of servers in applications such as multi-core processors and data centers has led to
a desire for low communication and memory requirements. On the other hand,
communication between the dispatcher and the servers, as well as memory at the dispatcher, allow the dispatcher to obtain and store information about the current state of the queues and about the characteristics of the servers, leading to better dispatching decisions. This points to a tradeoff between the resources utilized (in terms of communication overhead and memory), and the attainable delay performance and stability of the system.

In this paper, we consider a heterogeneous distributed service system, where servers can have different and unknown processing rates, and explore the tradeoff between the stability region of the system and the amount of communication overhead and memory used to gather and store relevant information. This complements the work in \cite{positiveResult,negativeResult}, where the authors explore the tradeoff between the delay performance and the amount of communication overhead and memory in a system with identical servers. In particular, in the setting of \cite{positiveResult,negativeResult} stability was easy to achieve (even with a static, randomized policy), and the focus was on the queueing delay going to zero (as the arrival rate and the number of servers jointly increase). In the present context, stability becomes an issue: the dispatcher must either ``learn" the rates of the different servers (and store this information in its memory), or must use some dynamic queue-size information to stabilize the system.

\subsection{Previous work}
There is a wide range of policies for operating the system described above, which result in different delay performances, stability regions, and resource utilizations. For example, a most simple policy is to dispatch jobs uniformly at random. This policy requires no message exchanges and no memory, but it is unstable if some server is slow enough. At the opposite extreme, the server can use dynamically available information and send incoming jobs to a shortest queue. This policy results in small delay and is maximally stable \cite{bramson}  but requires substantial communication overhead and an unbounded memory.

Many intermediate policies have been proposed and analyzed in the past, with a focus on low resource usage. Most notably, the Power-of-$d$-Choices (also known as SQ($d$)) was introduced and analyzed in \cite{mitzenmacher,vvedenskaya}, and results in relatively low average delays for the jobs, while requiring a message rate proportional to the arrival rate, and no memory. However, the blind randomization used by the policy renders it unstable if there is at least one sufficiently slow server. Another popular policy is Join-Idle-Queue \cite{joinIdleQueue,stolyar14}, which leverages the power of memory (one bit per server) to obtain vanishing queueing delays (as the arrival rate and the number of servers jointly increase) while using roughly the same amount of communication overhead as the Power-of-$d$-Choices. However, this policy also utilizes blind randomization that renders it unstable if there is at least one sufficiently slow server~\cite{persistentIdle}.

Recently, there has been a focus on policies that attain a vanishing queueing delay while minimizing their resource usage. In particular, in \cite{BorstPowerOfd} a variation of the Power-of-$d$-Choices was shown to yield a vanishing queueing delay while using no memory, and a message rate that is superlinear in the arrival rate. Moreover, variations of Join-Idle-Queue were shown to have vanishing queueing delays with either a memory of size (in bits) superlogarithmic in the number of servers and a message rate equal to the arrival rate \cite{positiveResult}, or a memory size (in bits) equal to the number of servers and a message rate strictly smaller (but still proportional) to the arrival rate \cite{Mark2020}. Last but not least, a novel combination of size-based load balancing and Round-Robin was shown to have vanishing queueing delay using unbounded memory and no communication overhead~\cite{Anselmi19}.

On the other hand, there are few policies in the literature that focus on maximizing the stability region. In \cite{loadBalancingWithMemory} the authors present and analyze a variation of Power-of-$d$-Choices that utilizes memory (of size logarithmic in the number of servers) to guarantee maximal stability. Furthermore, in \cite{persistentIdle} the authors propose yet another variation of Join-Idle-Queue, dubbed Persistent-Idle, that achieves maximal stability, without any randomization.
This policy requires a message rate proportional to the arrival rate, and a memory of size (in bits) at least proportional to the number of servers.

\subsection{Our contribution}
Instead of focusing on yet another policy or decision making architecture, we step back and address a more fundamental question: What are the message rate (from dispatcher to servers and from servers to dispatcher combined) and/or memory size requirements that are necessary and sufficient in order for a policy to be maximally stable? We are able to provide a fairly complete answer to this question.

\begin{itemize}
\item [a)] For the case of Poisson arrivals and exponential service times: If the message rate is positive {\bf and} the memory size (in bits) is logarithmic in the number of servers, we provide a fairly simple and natural policy that is maximally stable.
\item [b)] If the message rate  is sublinear in the square of the arrival rate {\bf and} the number of memory bits is sublogarithmic in the number of servers, we show that \emph{no} decision making architecture and policy, within a certain broad class of policies, is maximally stable. The main constraint that we impose on the policies that we consider is that they are ``weakly symmetric", in a sense to be defined later.
\end{itemize}

In a nutshell, as long as the message rate is not too excessive, a logarithmic memory is necessary and sufficient for maximal stability.

\begin{remark}
Our proposed policy is more economical than the most efficient maximally stable policy analyzed in earlier literature, the Power-of-$d$-Choices with memory policy \cite{loadBalancingWithMemory}, which requires a memory of size (in bits) at least logarithmic in the number of servers, and a message rate proportional to the arrival rate. In contrast, our proposed policy requires a memory size (in bits) logarithmic in the number of servers, and an arbitrarily small message rate.
\end{remark}




\section{Model and main results}\label{sec:results}
In this section, we present our modeling assumptions and main results. We present a unified framework for a broad set of dispatching policies, which includes most of the policies studied in the previous literature, and then present our negative result on the failure of maximal stability to hold for resource-constrained policies within this framework.

\subsection{Modeling assumptions}\label{sec:model_assumption}
We consider a distributed service system consisting of $n$ parallel servers, where each server is associated with an infinite capacity FIFO queue. For each  $i\in\{1,\dots,n\}$, the $i$-th server has constant (but unknown) service rate $\mu_i>0$. Despite the heterogeneity in the service rates, we assume that the total processing power of all servers is equal to $n$. Thus, the set of possible service rate vectors is
\begin{equation}\label{eq:serverRates}
  \Sigma_n \triangleq \left\{ {\bf \mu} \in (0,\infty)^n : \sum\limits_{i=1}^n \mu_i = n  \right\}.
\end{equation}

Jobs arrive to the system as a single renewal process of rate $\lambda n$ (for some fixed $\lambda\in(0,1)$), and their sizes are i.i.d., independent from the arrival process, and have a general distribution with unit mean. A central controller (dispatcher) is responsible for routing every incoming job to a queue, immediately upon arrival. We assume that the dispatcher can only rely on a limited amount of local memory and on messages that provide partial information about the state and parameters of the system. These messages (which are assumed to be instantaneous) can be sent from a server to the dispatcher at any time, or from the dispatcher to a server (in the form of queries) at the time of an arrival or at the time of a spontaneous message from a server. Messages from a server $i$ can only contain information about the state of its own queue (number of remaining jobs and the remaining workload of each one) and about its processing rate $\mu_i$. Within this context, a system designer has the freedom to choose a messaging policy, as well as the rules for updating the memory and for selecting the destination of an incoming job.\\

Regarding the performance metric, our focus is on the {\bf stability region} of a policy under the arrival rate $\lambda$, i.e., the largest subset of server rates $\Gamma_n(\lambda) \subset \Sigma_n$ such that the policy is stable for all ${\bf \mu} \in\Gamma_n(\lambda)$. We will formalize this definition in Subsection~\ref{sec:instability}.

\subsection{A maximally stable policy} \label{sec:stablePolicy}


In this subsection we propose a simple dispatching policy with the largest possible stability region (i.e., with stability region equal to $\Sigma_n$, for all $\lambda\in(0,1)$).

\subsubsection{Policy description}\label{sec:stablePolicyDescription}
For any fixed value of $n$, we consider the following policy. At any time, the dispatcher stores the ID of a single server in its memory.  This ID is initialized in an arbitrary way, and it is updated based on spontaneous messages from the servers. In particular, each server sends messages to the dispatcher as an independent Poisson process of rate $\alpha_n>0$, informing the dispatcher of its queue length (i.e., of the number of jobs in its queue or in service). When a message from a server arrives to the dispatcher, the dispatcher stores the ID of this server only if the sender's queue is shorter than the queue of the server that is currently stored in memory. In order to make this comparison, the queue length of the currently stored server is obtained by sending a query to it. Finally,  whenever a new job arrives to the system, it is sent to the server whose ID is stored in the dispatcher's memory (the server ID in memory does not change at this point).

\begin{remark}
This policy requires only $\lceil\log_2(n)\rceil$ bits of memory, and an arbitrarily small (but positive) average message rate of $3 \alpha_n n$.
\end{remark}

\subsubsection{Main result}\label{sec:stabilityResult}

When the arrival process is Poisson and the service times are exponentially distributed, the behavior of the system under this policy can be modeled as a continuous-time Markov chain $\big({\bf Q}(\cdot),I(\cdot)\big)$, where ${\bf Q}(\cdot)=\big( {\bf Q}_1(\cdot),\dots,{\bf Q}_n(\cdot) \big)$ is the vector of queue lengths and $I(\cdot)$ is the ID of the server stored in memory. In this setting, the stability of the policy is established in the following result.

\begin{theorem}\label{thm:stablePolicy}
  Suppose that the arrival process is Poisson, and that the job sizes are exponentially distributed. For any $n$, if $\alpha_n>0$, then the stability region of the policy described above is $\Sigma_n$, for all $\lambda\in(0,1)$.
\end{theorem}
This stability result is established by constructing an appropriate Lyapunov function. The proof is given in Appendix~\ref{sec:proof_stability}.\\

Theorem \ref{thm:stablePolicy} states that, at least in the Markovian case, the stability region of our proposed policy is the whole set of admissible rates $\Sigma_n$. Moreover, it implies that $\lceil \log_2(n) \rceil$ bits of memory and an average message rate of $3 \alpha_n n$ (which can be arbitrarily small) are sufficient for a policy to be always stable. We conjecture that our policy is always stable, even with renewal arrivals and generally distributed service times.

While Theorem \ref{thm:stablePolicy} ensures stability even when the message rate is arbitrarily small, a small message rate will result in poor steady-state delay performance of the policy. Indeed, since the policy sends all incoming jobs to the same queue between consecutive messages, a small message rate leads to large build ups of jobs in the queues. In particular, we conjecture that the steady-state queueing delay of our policy is of order $\Theta(1/\alpha_n)$.

Given this apparent tradeoff between the average message rate and the delay performance of the policy, the proposed policy is most useful for applications where a large stability region and a small communication overhead are preferred, and where there is tolerance for large delays. Furthermore, since the policy does not depend explicitly on the rates of the servers (or estimates thereof), it would continue to work even if the service rates changed slowly over time, which makes it robust.

\subsection{A general class of dispatching policies} \label{sec:unifiedFramework2}

In this subsection we present a unified framework that describes memory-based dispatching policies in systems with heterogeneous servers, which is slightly more general than the one introduced in \cite{negativeResult}.

Let $c_n$ be the number of memory bits available to the dispatcher. We define the corresponding set of memory states to be $\mathcal{M}_n \triangleq \left\{1,\dots,2^{c_n}\right\}$. 
Furthermore, we define the set of possible states at a server as the set of nonnegative sequences $\mathcal{Q}\triangleq\mathbb{R}_+^{\mathbb{Z}_+}$, where a sequence specifies the remaining workload of each job in that queue, including the one that is being served. (In particular, an idle server is represented by the zero sequence.) As long as a queue has a finite number of jobs, the queue state is a sequence that has only a finite number of non-zero entries. The reason that we include the workload of the jobs in the state is that we wish to allow for a broad class of policies that can take into account the remaining workload in the queues.
In particular, we allow for information-rich messages that describe the full workload sequence at the server that sends the message. We are interested in the process
\[ {\bf Q}(\cdot)= \big({\bf Q}_1(\cdot),\dots,{\bf Q}_n(\cdot)\big) = \Big(\big({\bf Q}_{1,j}(\cdot)\big)_{j=1}^\infty,\dots,\big({\bf Q}_{n,j}(\cdot)\big)_{j=1}^\infty\Big),\]
which takes values in the set $\mathcal{Q}^n$, and describes the evolution of the workload of each job in each queue. We are also interested in the process $M(\cdot)$ that describes the evolution of the memory state, and in a process $Z(\cdot)$ that describes the elapsed time since the arrival of the previous job.

\subsubsection{Fundamental processes and initial conditions}\label{sec:fund}
All processes of interest are driven by the following common fundamental processes:
\begin{enumerate}
  \item {\bf Arrival process:} A delayed renewal counting process $A_n(\cdot)$ with rate $\lambda n$, and event times $\{T_k\}_{k=1}^\infty$, defined on a probability space $(\Omega_A,\mathcal{A}_A,\mathbb{P}_A)$.
  \item {\bf Spontaneous messages process:} A Poisson counting process $R_n(\cdot)$ with rate $\beta_n$, and event times $\{T^s_k\}_{k=1}^\infty$, defined on a probability space $(\Omega_R,\mathcal{A}_R,\mathbb{P}_R)$.
  \item {\bf Job sizes:} A sequence of i.i.d. random variables $\{W_k\}_{k=1}^\infty$ with mean one, defined on a probability space $(\Omega_W,\mathcal{A}_W,\mathbb{P}_W)$.
  \item {\bf Randomization variables:} Eight independent and individually i.i.d. sequences of random variables $\{U_{1,k}\}_{k=1}^\infty,\dots,\{U_{8,k}\}_{k=1}^\infty$, uniform on $[0,1]$, defined on a common probability space $(\Omega_U,\mathcal{A}_U,\mathbb{P}_U)$.
  \item {\bf Initial conditions:} Random variables ${\bf Q}(0)$, $M(0)$, and $Z(0)$, defined on a common probability space $(\Omega_0,\mathcal{A}_0,\mathbb{P}_0)$.
\end{enumerate}
The whole system will be defined on the associated product probability space
\[ \big(\Omega_A\times\Omega_R\times\Omega_W\times\Omega_U\times\Omega_0, \mathcal{A}_A\times\mathcal{A}_R\times\mathcal{A}_W\times\mathcal{A}_U\times\mathcal{A}_0, \mathbb{P}_A\times\mathbb{P}_R\times\mathbb{P}_W\times\mathbb{P}_U\times\mathbb{P}_0\big), \]
to be denoted by $(\Omega,\mathcal{A},\mathbb{P})$. All of the randomness in the system originates from these fundamental processes, and everything else is a deterministic function of them.

\subsubsection{A construction of sample paths}\label{sec:samplePaths2}
We provide a construction of a Markov process $({\bf Q}(\cdot),M(\cdot),Z(\cdot))$, taking values in the set $\mathcal{Q}^n\times\mathcal{M}_n\times\mathbb{R}_+$. The memory process $M(\cdot)$ is piecewise constant, and can only jump at the time of an event. All processes considered will have the c\`adl\`ag property (right-continuous with left limits) either by assumption (e.g., the underlying fundamental processes) or by construction.

There are three types of events: job arrivals, spontaneous messages, and service completions. We now describe the sources of these events, and what happens when they occur.\\

\noindent{\bf Job arrivals:} At the time of the $k$-th event of the arrival process $A_n$, which occurs at time $T_k$ and involves a job with size $W_k$, the following transitions happen sequentially but instantaneously:
\begin{enumerate}
  \item First, the dispatcher chooses a set ${S}_k$ of distinct servers, from which it solicits information about their state, according to
  \[ {S}_k=f_1\Big(M\left(T_k^-\right),W_k,U_{1,k}\Big), \]
  where $f_1:\mathcal{M}_n\times\mathbb{R}_+\times[0,1]\to\mathcal{P}(\{1,\dots,n\})$ is a measurable function defined by the policy. Here, and in the sequel, $\mathcal{P}(A)$ stands for the power set of a set $A$.
  \item Second, messages are sent to the servers in the set ${S}_k$, and the servers respond with messages containing their queue states and their service rates. This results in a total of $2|{S}_k|$ messages exchanged. Using this information, the destination of the incoming job is chosen to be
  \[ D_k=f_2\Big(M\big(T_k^-\big),W_k,\Big\{\Big({\bf Q}_{i}\big(T_k^-\big),{\bf \mu}_{i},i\Big):i\in S_k\Big\},U_{2,k}\Big), \]
  where $f_2:\mathcal{M}_n \times\mathbb{R}_+ \times \mathcal{B}_n \times[0,1]\to\{1,\dots,n\}$ is a measurable function defined by the policy,  with $\mathcal{B}_n\subset \mathcal{P}\big(\mathcal{Q}\times\mathbb{R}_+\times\{1,\dots,n\}\big)$
  comprised of those sets of triples such that  the triples in a set have different third coordinates. Note that the destination of a job can depend not only on the current memory state, the job size, the set of queried servers, and the state of their queues, but also on the rates of the queried servers.
  \item Third, the memory state is updated according to
  \[ M(T_k)=f_3\Big(M\big(T_k^-\big),W_k,\Big\{\Big({\bf Q}_{i}\big(T_k^-\big),{\bf \mu}_{i},i\Big):i\in S_k\Big\},D_k,U_{3,k}\Big), \]
  where $f_3:\mathcal{M}_n\times\mathbb{R}_+\times \mathcal{B}_n \times\{1,\dots,n\}\times[0,1]\to\mathcal{M}_n$ is a measurable function defined by the policy. Note that the new memory state is obtained using the same information as for selecting the destination, including the rates of the queried servers, plus the destination of the job.
\end{enumerate}

\noindent{\bf Spontaneous messages:} At the time of the $k$-th event of the spontaneous message process $R_n$, which occurs at time $T^s_k$, the $i$-th server sends a spontaneous message to the dispatcher if and only if
\[ g_1\Big({\bf Q}\big(T^s_k\big),{\bf \mu},U_{4,k}\Big)=i, \]
where $g_1:\mathcal{Q}^n\times\mathbb{R}_+^n\times[0,1]\to\{0,1,\dots,n\}$ is a measurable function defined by the policy. On the other hand, no message is sent when $g_1\big({\bf Q}({T^s_k}),{\bf \mu},U_{4,k}\big)=0$. Note that the dependence of $g_1$ on ${\bf Q}$ and ${\bf \mu}$ allows the message rate at each server to depend on all servers' current workloads, and on their rates. This allows for policies that let servers with higher service rates send messages at a higher rate than servers with slower service rates.

When a spontaneous message from server $i$ arrives to the dispatcher, the following transitions happen sequentially but instantaneously:
\begin{enumerate}
\item First, the dispatcher chooses a set of distinct servers ${S}^s_k$, from which it solicits information about their state, according to
  \[ {S}^s_k=g_2\Big(M\left(T_k^-\right), i, {\bf Q}_i\big(T^s_k\big), \mu_i,U_{5,k}\Big), \]
  where $g_2:\mathcal{M}_n\times\{1,\dots,n\}\times\mathcal{Q}\times\mathbb{R}_+\times[0,1] \to\mathcal{P}(\{1,\dots,n\})$ is a measurable function defined by the policy. Note that the set of servers that are sampled not only depends on the current memory state but also on the index, queue state, and rate of the server that sent the message.
\item Second, messages are sent to the servers in the set ${S}^s_k$, and the servers respond with messages containing their queue states and their service rates. This results in a total of $2|{S}^s_k|$ messages exchanged. Using this information, the memory is updated to the new memory state
\[ M(T^s_k) = g_3\Big(M\big({T^s_k}^-\big), i, {\bf Q}_i\big(T^s_k\big), \mu_i,\Big\{\Big({\bf Q}_{j}\big(T_k^s\big),{\bf \mu}_{j},j\Big):j\in S^s_k\Big\},U_{6,k}\Big), \]
where $g_3: \mathcal{M}_n\times \{1,\dots,n\} \times \mathcal{Q}\times\mathbb{R}_+\times \mathcal{B}_n\times[0,1]
\to\mathcal{M}_n$ is a measurable function defined by the policy.
\end{enumerate}

\noindent{\bf Service completions:} Let $\{T^d_k(i)\}_{k=1}^\infty$ be the sequence of departure times at the $i$-th server. At those times, the $i$-th server sends a message to the dispatcher if and only if
\[ h_1\Big({\bf Q}_i\big({T^d_k(i)}\big),\mu_i,U_{7,k}\Big)=1, \]
where $h_1:\mathcal{Q}\times\mathbb{R}_+\times[0,1]\to\{0,1\}$ is a measurable function defined by the policy. In that case, the memory is updated to the new memory state
\[ M\Big(T^d_k(i)\Big) = h_2\Big(M\big({T^d_k(i)}^-\big), i, {\bf Q}_i\big(T^d_k(i)\big),\mu_i,U_{8,k}\Big), \]
where $h_2:\mathcal{M}_n\times\{1,\dots,n\}\times\mathcal{Q}\times\mathbb{R}_+\times[0,1]\to\mathcal{M}_n$ is a measurable function defined by the policy. Finally, no message is sent when $h_1\big({\bf Q}_i({T^d_k(i)}),\mu_i,U_{7,k}\big)=0$.

\begin{remark}
Note that this framework allows for policies that  are more general than those considered in \cite{negativeResult}. In particular, (i) some decisions can depend on the rates of the different servers, (ii) the dispatcher can sample servers whenever a spontaneous message arrives, and (iii) memory updates may involve randomization.
\end{remark}

We now introduce a symmetry assumption on the policies.

\begin{assumption}(Weakly symmetric policies.)\label{def:weakSymmetry}
We assume that the dispatching policy is weakly symmetric, in the following sense. For any given permutation of the servers $\sigma$, there exists a corresponding (not necessarily unique) permutation $\sigma_M$ of the memory states $\mathcal{M}_n$ that satisfies both of the following properties:
\begin{enumerate}
  \item For every $m\in\mathcal{M}_n$ and $w \in \mathbb{R}_+$, and if $U$ is a uniform random variable on $[0,1]$, then
      \[ \sigma\Big(f_1(m,w,U)\Big) \overset{d}{=} f_1\big(\sigma_M(m),w,U\big), \]
      where $\overset{d}{=}$ stands for equality in distribution. Note that this equality in distribution is only with respect to $U$.
  \item For every $m\in\mathcal{M}_n$, $w \in \mathbb{R}_+$, ${S}\in\mathcal{P}(\{1,\dots,n\})$, ${\bf q}\in\mathcal{Q}^{n}$, and ${\bf \mu}\in\mathbb{R}_+^{n}$, and if $V$ is a uniform random variable on $[0,1]$, then
   \begin{align*}
   &\sigma\Big(f_2\big(m,w,\big\{({\bf q}_i,{\bf \mu}_i,i):i\in S\big\},V\big)\Big) \\
   &\qquad\qquad\qquad\qquad\qquad \overset{d}{=} f_2\Big(\sigma_M(m),w,\Big\{\big({\bf q}_i,{\bf \mu}_i,\sigma(i)\big):i\in S\Big\},V\Big).
  \end{align*}
\end{enumerate}
\end{assumption}

\begin{remark}
 This assumption prevents any bias for or against a server, unless it is encoded in the memory in a sufficiently detailed way so that the assumption is satisfied. For example, in order to implement (in a weakly symmetric way) the randomized dispatching policy where incoming jobs are sent to a server with a probability proportional to its processing rate, the second condition in Assumption~\ref{def:weakSymmetry} requires the dispatching probabilities  to be encoded in memory, in a sufficiently detailed way.
\end{remark}

%

\begin{remark}
  Note that the universally stable policy introduced in Subsection~\ref{sec:stablePolicyDescription} falls within the class of policies defined by this general framework, and  it satisfies Assumption~\ref{def:weakSymmetry}.
\end{remark}

\subsection{Instability of resource constrained policies} \label{sec:instability}
In this subsection we state the main result about the instability of general weakly symmetric dispatching policies. Before stating this main result, we first define the {\bf average message rate} between the dispatcher and the servers as
\begin{align}
 \limsup_{t\to\infty} \frac{1}{t}&\left[ \sum\limits_{k=1}^{A_n(t)} 2|{S}_k| + \sum\limits_{k=1}^{R_n(t)} \Big(1+2|{S}^s_k|\Big) \mathds{1}_{\{1,\dots,n\}}\Big(g_1\big({\bf Q}\big({T^s_k}\big),{\bf \mu},U_{4,k}\big)\Big) \right. \nonumber \\
  &\qquad\qquad\qquad\qquad \left. + \sum\limits_{i=1}^n \sum\limits_{k:\, T^d_k(i)<t}
   \mathds{1}_{\{1\}}\Big(h_1\big({\bf Q}_i\big({T^d_k(i)}\big),{\bf \mu}_i,U_{7,k}\big)\Big)\right]. \label{eq:messageRate2}
\end{align}

Second, we provide a formal definition of our performance metric: the stability region of a policy. For each $n$, given a policy and an arrival rate $\lambda$, the {\bf stability region} of the policy under the arrival rate $\lambda$, denoted by $\Gamma_n(\lambda)$, is the set of all server rates in $\Sigma_n$
for which the process $\big({\bf Q}(\cdot),M(\cdot),Z(\cdot)\big)$ is positive Harris recurrent.\\

We are now ready to state our main negative result. It asserts that within the class of weakly symmetric policies that we consider, and under some constraints on the memory size and the message rate, the stability region
does not contain all possible rates.

\begin{theorem}[Instability of resource constrained policies]
\label{thm:instability}
For any constant $\lambda\in(0,1)$ and positive sequence $\{\alpha_n\}_{n\geq 1}$, there exists a sequence of stability regions $\big\{\Gamma_n(\lambda,\alpha_n)\big\}_{n\geq 1}$, where $\Gamma_n(\lambda,\alpha_n)\subsetneq \Sigma_n$ for all $n\geq 1$, with the following property.

Consider a sequence of weakly symmetric memory-based dispatching policies, i.e., that satisfy Assumption~\ref{def:weakSymmetry}, with at most $c_n\in o\big(\log (n)\big)$ bits of memory, and with an average message rate (cf. Equation~\ref{eq:messageRate2}) upper bounded by $\alpha_n \in o\big(n^2\big)$ almost surely. Then, for all $n$ large enough, the stability region of the policies under the arrival rate $\lambda$ are contained in $\Gamma_n(\lambda,\alpha_n)$.
\end{theorem}

The proof consists of showing that, when half of the servers have a sufficiently small service rate $\epsilon_n\in \Theta\big(e^{-\alpha_n/n}\big)$, the total workload of the system diverges, as a function of time, for all $n$ large enough.
It also relies heavily on a combinatorial result (Proposition  \ref{prop:symSampling}) from \cite{negativeResult}, on the limitations imposed by symmetry on a limited memory. The proof is given in Appendix~\ref{sec:proof_instability}.

\begin{remark}
Theorem~\ref{thm:instability} states that the stability region of a policy is contained in a proper subset of $\Sigma_n$, which only depends on $n$, the arrival rate $\lambda$, and on the message rate $\alpha_n$. This means that, for all $n$ large enough, and as long as $\alpha_n \in o\big(n^2\big)$, there is at least one vector of processing rates for which the system is unstable under any weakly symmetric policy with $o(\log(n))$ bits of memory, and a message rate that is upper bounded by $\alpha_n$.
\end{remark}

\begin{remark}
The most interesting regime is the one where $\alpha_n\in O(n)$, that is, when we have a constant number of messages per job. In this regime, when half of the servers have rate $\epsilon_n\in\Theta(1)$, weakly symmetric policies with $o\big(\log (n)\big)$ memory are unstable. In particular, resource constrained policies become unstable for a significant portion of the possible service rate vectors~$\Sigma_n$.

  On the other hand, when $\alpha_n\in\omega(n)$, our result requires half of the servers to have rate $\epsilon_n\in\Theta\big(e^{-\alpha_n/n}\big)$, which is exponentially small in $\alpha_n/n\in\omega(1)$. This suggests that, when the average message grows faster than $n$, it is only a very small portion of the possible service rate vectors that can destabilize all resource constrained policies.
\end{remark}

\subsection{Stability versus resources tradeoff}\label{sec:stabilitySummary}
In this subsection, we provide a visual summary of our results on the tradeoff between the stability region, and the memory and communication overhead of weakly symmetric dispatching policies.

First, according to Theorem~\ref{thm:stablePolicy}, with a memory size of at least $\lceil\log_2(n)\rceil$ bits and with an arbitrarily small message rate, we can obtain a weakly symmetric policy that is always stable (for any service rate vector in $\Sigma_n$). Second, Theorem~\ref{thm:instability} states that weakly symmetric policies with $o\big(\log(n)\big)$ bits of memory and a message rate of order $o\big(n^2\big)$ cannot be always stable. Finally, note that both the Join-Shortest-Queue policy, and a policy which sends incoming jobs to each server with a probability proportional to the server's rate, can be implemented by querying all servers at the time of each arrival. These policies require a message rate of order $\Theta\big(n^2\big)$, and no memory, and they are always stable. The three regimes are depicted in Figure~\ref{fig:stabilityRegimes}.

\begin{figure}[ht!]
\begin{center}
\begin{tikzpicture}[scale=1.3]
  \filldraw[black!10] (5,3.5) -- (5,4) -- (5.5,4) -- (5.5,3.5) -- (5,3.5);
  \filldraw[white] (2.5,0) -- (2.5,2) -- (5,2) -- (5,0) -- (2.5,0);
  \draw (5,3.5) -- (5,4) -- (5.5,4) -- (5.5,3.5) -- (5,3.5);
  \draw (6.85,3.75) node {Not always stable};
  \draw (5,2.9) -- (5,3.4) -- (5.5,3.4) -- (5.5,2.9) -- (5,2.9);
  \draw (6.6,3.15) node {Always stable};

  \filldraw[black!10] (0,0) -- (0,2) -- (2.5,2) -- (2.5,0) -- (0,0);

  \draw [->,thick] (-0.5,0) -- (5.1,0);
  \draw (6.6,0) node {Total message rate};
  \draw [->,thick] (0,-0.5) -- (0,4);
  \draw (0,4.3) node {Bits of memory};

  \draw (1.25,-0.3) node {$o\big(n^2\big)$};
  \draw (3.75,-0.3) node {$\Omega\big(n^2\big)$};

  \draw (-0.8,3) node {$\Omega(\log(n))$};
  \draw (-0.8,1) node {$o(\log(n))$};

  \draw (0,2) -- (5,2);
  \draw (2.5,0) -- (2.5,2);

  \draw (2.5,3) node {Theorem~\ref{thm:stablePolicy}};

  \draw (1.25,1) node {Theorem~\ref{thm:instability}};

  \draw (3.75,1.25) node {Weighted random};
  \draw (3.75,0.75) node {policy};

\end{tikzpicture}
\end{center}
\caption{Resource requirements for stable policies.}
\label{fig:stabilityRegimes}
\end{figure}
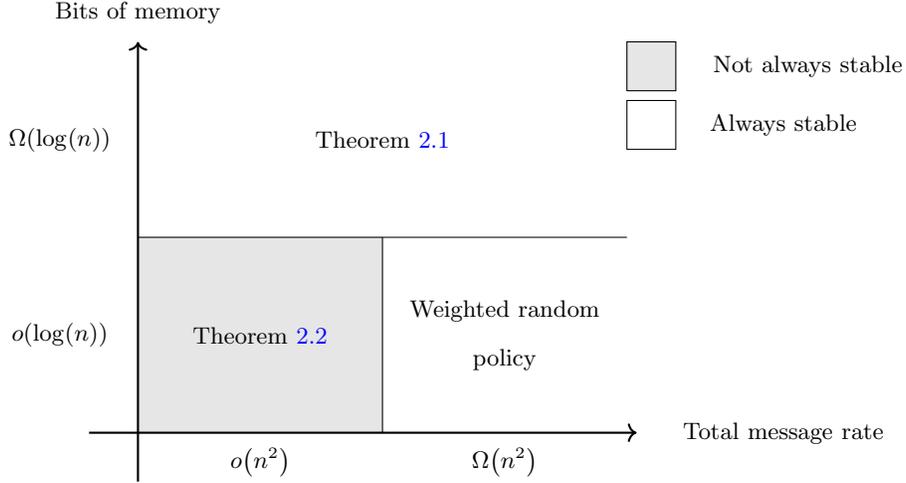

The only remaining question in this setting is whether stability can be guaranteed with zero communication overhead, and $\Omega(\log(n)$ bits of memory. In this case, no messages are exchanged, and the dispatcher can never obtain information about the rate of the servers. As a result, it can only dispatch jobs blindly, and stability fails for some server rates.

\section{Conclusions and future work}\label{sec:conclusions}

In this paper, we proposed a simple but efficient dispatching policy that requires a memory of size (in bits) logarithmic in the number of servers, and an arbitrarily small message rate, and showed that it has the largest possible stability region. The key to the stability properties of this policy is the fact that it never chooses the destination of a job by random sampling of the servers (like the Power-of-$d$-Choices) or by random dispatching of the job (like Join-Idle-Queue). On the other hand, we showed that when we have a memory size (in bits) sublogarithmic in the number of servers, and a message rate sublinear in the square of the arrival rate, all weakly symmetric dispatching policies have a sub-optimal stability region.

There are several interesting directions for future research. For example:
\begin{itemize}
\item [(i)] While policies can have the largest possible stability region using an arbitrarily small message rate and logarithmic memory, their delay performance could be arbitrarily bad. We conjecture that the average delay of a policy is at least inversely proportional to its average message rate per server.
\item [(ii)] In light of the symmetry assumption in Theorem \ref{thm:instability}, a natural question is whether the result still holds without it. In that case, the sampling of servers and dispatching of jobs need not be uniform (as established in propositions \ref{prop:symSampling} and \ref{prop:symDispatching} using the symmetry assumption), and it becomes unclear whether maximal stability is still impossible in the same regime.
\end{itemize}

\appendix

\section{Proof of Theorem 2.1}\label{sec:proof_stability}
Let us fix some $n$ and some arbitrary vector of processing rates in $\Sigma_n$. Let $\mu_{min}$ and $\mu_{max}$ be the smallest and largest processing rates in the chosen vector, respectively. In particular, note that they are positive.

We will use the Foster-Lyapunov criterion to show that the continuous-time Markov chain $\big({\bf Q}(\cdot),I(\cdot)\big)$ is positive recurrent. First, note that this process has state space $\mathbb{Z}_+^n\times\{1,\dots,n\}$. Its transition rates, denoted by $r_{\cdot\,\to\,\cdot}$, are as follows, where we use ${\bf e}_j$ to denote the $j$-th unit vector in $\mathbb{Z}^n_+$:
\begin{enumerate}
  \item Since incoming jobs are sent to the queue whose ID is stored in memory, each queue sees arrivals with rate:
      \[ r_{({\bf q},i)\to({\bf q}+{\bf e}_i,i)} = \lambda n. \]
  \item Transitions due to service completions occur according to the processing rate of each server, and they do not affect the ID stored in memory:
      \[ r_{({\bf q},i)\to({\bf q}-e_j,i)} = \mu_j \mathds{1}_{[1,\infty)}\big({\bf q}_j\big). \]
  \item Spontaneous messages are sent from each server to the dispatcher at a rate equal to $\alpha_n$, but the ID stored in memory only changes if the sender of the message has a shorter queue:
      \[ r_{({\bf q},i)\to({\bf q},j)} = \alpha_n \mathds{1}_{[0,{\bf q}_i-1]}\big({\bf q}_{j}\big). \]
      \item Any transitions that do not appear in the above have zero rate.
\end{enumerate}
Note that the Markov process $\big({\bf Q}(\cdot),I(\cdot)\big)$ on the state space $\mathbb{Z}_+^n\times\{1,\dots,n\}$ is irreducible, with all states reachable from each other. To show positive recurrence, we define the Lyapunov functions
\[ \Xi_1({\bf q},i) \triangleq \frac{2\mu_{max}}{\alpha_n}{\bf q}_i, \]
\[ \Xi_2({\bf q},i) \triangleq \sum\limits_{j=1}^n {\bf q}_j^2, \]
and
\begin{equation}\label{eq:fullLyapunov}
 \Xi({\bf q},i) \triangleq \Xi_1({\bf q},i) + \Xi_2({\bf q},i),
\end{equation}
and note that
\[ \sum\limits_{({\bf q}',i')\neq({\bf q},i)} \Xi({\bf q}',i')r_{({\bf q},i)\to({\bf q}',i')} < \infty,  \quad\quad \forall \, ({\bf q},i)\in \mathbb{Z}_+^n\times\{1,\dots,n\}.  \]
We also define the finite set
\begin{equation}\label{eq:smallSetStability}
 F_n\triangleq \left\{({\bf q},i)\in \mathbb{Z}_+^n\times\{1,\dots,n\} : \sum\limits_{j=1}^n {\bf q}_j < \frac{\lambda n \left( 1+\frac{2\mu_{max}}{\alpha_n}\right) + n + 1}{2\min\{1-\lambda,\mu_{min}\}} \right\}.
\end{equation}
For any state $({\bf q},i)$, we have
  \begin{align}
    &\sum\limits_{({\bf q}',i')\in \mathbb{Z}_+^n\times\{1,\dots,n\}}  \Big[\Xi_1({\bf q}',i')-\Xi_1({\bf q},i)\Big]r_{({\bf q},i)\to({\bf q}',i')} \nonumber \\
     &\qquad\qquad\quad = \lambda n \left( \frac{2\mu_{max}}{\alpha_n}\right)  - \frac{2\mu_{max}}{\alpha_n}\mu_i\mathds{1}_{[1,\infty)}\big({\bf q}_i\big)
     - \sum\limits_{j=1}^n 2\mu_{max}\big({\bf q}_i-{\bf q}_j\big)^+
     \nonumber \\
    &\qquad\qquad\quad \leq \lambda n \left( \frac{2\mu_{max}}{\alpha_n}\right) - \sum\limits_{j=1}^n 2\mu_{max}\big({\bf q}_i-{\bf q}_j\big)^+, \label{eq:xi1}
  \end{align}
and
\begin{align}
    &\sum\limits_{({\bf q}',i')\in \mathbb{Z}_+^n\times\{1,\dots,n\}}  \Big[\Xi_2({\bf q}',i')-\Xi_2({\bf q},i)\Big]r_{({\bf q},i)\to({\bf q}',i')} \nonumber \\
     &\qquad\qquad\qquad\qquad\qquad = \lambda n \left( 2{\bf q}_i + 1\right)  - \sum\limits_{j=1}^n \mu_j\big(2{\bf q}_j-1\big) \mathds{1}_{[1,\infty)}\big({\bf q}_j\big) \nonumber \\
    &\qquad\qquad\qquad\qquad\qquad = \lambda n \left( 2{\bf q}_i + 1\right) + \sum\limits_{j=1}^n \mu_j \mathds{1}_{[1,\infty)}\big({\bf q}_j\big) - 2\sum\limits_{j=1}^n \mu_j {\bf q}_j  \nonumber \\
    &\qquad\qquad\qquad\qquad\qquad \leq \lambda n \left( 2{\bf q}_i + 1\right) + n - 2\sum\limits_{j=1}^n \mu_j {\bf q}_j, \label{eq:xi2}
\end{align}
where in the last inequality we used that the vector of server rates $\mu$ is in $\Sigma_n$, so that
\begin{equation}\label{eq:properSum}
 \sum\limits_{j=1}^n \mu_j = n.
\end{equation}

Combining equations \eqref{eq:fullLyapunov}, \eqref{eq:xi1}, and \eqref{eq:xi2}, for any state $({\bf q},i)\notin F_n$, we have
\begin{align*}
    &\sum\limits_{({\bf q}',i')\in \mathbb{Z}_+^n\times\{1,\dots,n\}}  \Big[\Xi({\bf q}',i')-\Xi({\bf q},i)\Big]r_{({\bf q},i)\to({\bf q}',i')} \\
    &\qquad\quad \leq \lambda n \left( 1+\frac{2\mu_{max}}{\alpha_n}\right) + n + 2\lambda n {\bf q}_i - 2\sum\limits_{j=1}^n \Big[ \mu_j {\bf q}_j + \mu_{max}({\bf q}_i-{\bf q}_j)^+ \Big] \\
    &\qquad\quad\leq \lambda n \left( 1+\frac{2\mu_{max}}{\alpha_n}\right) + n
    + 2\lambda n {\bf q}_i - 2\sum\limits_{j=1}^n \mu_j \Big[ {\bf q}_j + ({\bf q}_i-{\bf q}_j)^+ \Big] \\
    &\qquad\quad= \lambda n \left( 1+\frac{2\mu_{max}}{\alpha_n}\right) + n
    + 2\lambda n {\bf q}_i - 2\sum\limits_{j=1}^n \mu_j \max\big\{ {\bf q}_i, \, {\bf q}_j \big\} \\
    &\qquad\quad= \lambda n \left( 1+\frac{2\mu_{max}}{\alpha_n}\right) + n
    + 2\lambda n {\bf q}_i - 2\sum\limits_{j=1}^n \mu_j \Big[ {\bf q}_i + ({\bf q}_j-{\bf q}_i)^+ \Big] \\
    &\qquad\quad= \lambda n \left( 1+\frac{2\mu_{max}}{\alpha_n}\right) + n
    + 2\lambda n {\bf q}_i - 2{\bf q}_i \sum\limits_{j=1}^n \mu_j  - 2\sum\limits_{j=1}^n \mu_j ({\bf q}_j-{\bf q}_i)^+ \\
    &\qquad\quad\overset{(*)}{=} \lambda n \left( 1+\frac{2\mu_{max}}{\alpha_n}\right) + n
    - 2(1-\lambda) n {\bf q}_i - 2\sum\limits_{j=1}^n \mu_j ({\bf q}_j-{\bf q}_i)^+ \\
    &\qquad\quad\leq \lambda n \left( 1+\frac{2\mu_{max}}{\alpha_n}\right) + n
    - 2(1-\lambda) n {\bf q}_i - 2\mu_{min}\sum\limits_{j=1}^n ({\bf q}_j-{\bf q}_i)^+ \\
    &\qquad\quad\leq \lambda n \left( 1+\frac{2\mu_{max}}{\alpha_n}\right) + n
    - 2\min\{1-\lambda,\mu_{min}\} \sum\limits_{j=1}^n \Big[{\bf q}_i+ ({\bf q}_j-{\bf q}_i)^+ \Big]\\
    &\qquad\quad= \lambda n \left( 1+\frac{2\mu_{max}}{\alpha_n}\right) + n
    - 2\min\{1-\lambda,\mu_{min}\} \sum\limits_{j=1}^n \max\big\{ {\bf q}_i, \, {\bf q}_j \big\} \\
    &\qquad\quad\leq \lambda n \left( 1+\frac{2\mu_{max}}{\alpha_n}\right) + n
    - 2\min\{1-\lambda,\mu_{min}\} \sum\limits_{j=1}^n {\bf q}_j\\
    &\qquad\quad\leq -1,
\end{align*}
where in equality $(*)$ we used Equation \eqref{eq:properSum}, and in the last inequality we used the fact that $({\bf q},i)\notin F_n$ and the definition of the finite set $F_n$ (Equation \eqref{eq:smallSetStability}). Then, the Foster-Lyapunov criterion \cite{fosterLyapunov} implies the positive recurrence of the Markov chain $\big({\bf Q}(\cdot),I(\cdot)\big)$. Finally, since this is true for all server rates in $\Sigma_n$, we conclude that $\Sigma_n$ is the stability region of the policy.

\section{Proof of Theorem 2.2}\label{sec:proof_instability}

Fix $\lambda$, and consider a vector of server rates in $\Sigma_n$ where $\lfloor n/2 \rfloor$ servers have rate $\epsilon_n>0$. We will show that, for any given $\lambda$, and for all $\epsilon_n$ small enough, every resource constrained dispatching policy that is weakly symmetric (i.e., satisfies Assumption~\ref{def:weakSymmetry}) overloads the slow servers.

The high-level outline of the proof is as follows. In Subsection~\ref{sec:limitation(i)} we show that under our weak symmetry assumption, the constraint on the number of bits available implies that the dispatcher treats all servers in a symmetric way, in some appropriate sense.

Then, in Subsection~\ref{subsec:highArrivalRate} we combine the results obtained in Subsection~\ref{sec:limitation(i)} with the bound on the average message rate to show that jobs are sent to slow servers (i.e., to servers with service rate $\epsilon_n$) with a positive rate that is bounded away from zero. This implies that the total workload of the servers diverges when $\epsilon_n$ is small enough, thus completing the proof.

In the proof that follows, we assume that the sequences $c_n$ (memory size) and $\alpha_n$ (message rate) have been fixed, and are of order $o(\log n)$ and $o(n^2)$, respectively.

\subsection{Local limitations of symmetry and finite memory} \label{sec:limitation(i)}
In this subsection we will show how the constraint of having only $o(\log(n))$ bits of memory affects the distribution of the sampled servers, and the distribution of the dispatched jobs. The results that we provide are corollaries or special cases of results in \cite{negativeResult}.\\

We first note that if the dispatcher has $o(\log(n))$ bits of memory, and if $n$ is large enough, then the distribution of the sampled servers is uniform among all sets of the same size.

\begin{proposition}\label{prop:symSampling}
Let $U$ be a uniform random variable over $[0,1]$. For all $n$ large enough, for every memory state $m\in\mathcal{M}_n$ and  every possible job size $w \in\mathbb{R}_+$, the following holds. Consider any set of servers $S\in\mathcal{P}(\{1,\dots,n\})$ with $|S|\in o(n)$. Consider the event
\[ B(m,w;S) = \big\{ f_1(m,w,U)=S\cup \{i\},\ \mbox{for some $i\notin S$}\big\},\]
and assume that the conditional probability measure
\[ \mathbb{P}\left(\,\, \cdot \,\, \big| \ B(m,w;S)\right) \]
is well-defined. Then,
\begin{equation*}
 \mathbb{P}\left(j\in f_1(m,w,U) \ \big| \ B(m,w;S)\right)
 \end{equation*}
 is the same for all $j \notin S$.
\end{proposition}
\begin{proof}
This is a special case of Proposition 5.1 in \cite{negativeResult}, noting that while the statement of that proposition requires $|S|\leq\sqrt{n}$, its proof
goes through under the weaker assumption $|S|\in o(n)$.
\end{proof}

\begin{corollary}\label{cor:symSampling}
Let $U$ be a uniform random variable over $[0,1]$. For all $n$ large enough, for every memory state $m\in\mathcal{M}_n$, for every possible job size $w \in\mathbb{R}_+$, and for any set of servers $S\in\mathcal{P}(\{1,\dots,n\})$ with $|S|\in o(n)$, we have
\[ \mathbb{P}\Big(f_1\big(m,w,U\big)=S\Big) = \mathbb{P}\Big(f_1\big(m,w,U\big)=\sigma(S)\Big), \]
for every permutation $\sigma$.
\end{corollary}
\begin{proof}
In order to simplify notation, we omit the dependence of $f_1$ on $m$ and $w$.
Let us fix a set $S$, and a transposition $\tau$. If $\tau(S)=S$, then it is trivially true that
\[ \mathbb{P}\big(f_1(U)=S\big) = \mathbb{P}\big(f_1(U)=\tau(S)\big). \]
On the other hand, if $\tau(S)\neq S$, then there exists some $i\in S$ such that $\tau(i)\notin S$. In that case, we have:
\begin{align*}
  \mathbb{P}\big(f_1(U)=S\big)  & = \mathbb{P}\big(f_1(U)=S \,\big|\, |f_1(U)|=|S| \big)\mathbb{P}\big( |f_1(U)|=|S| \big) \\
  & =\mathbb{P}\Big(i\in f_1(U) \,\Big|\, \big\{ |f_1(U)|=|S| \big\} \cap \big\{ S\backslash\{i\} \subset f_1(U) \big\} \Big) \\
   &\qquad\qquad \cdot\mathbb{P}\Big(S\backslash\{i\} \subset f_1(U) \,\Big|\, |f_1(U)|=|S| \Big) \mathbb{P}\big( |f_1(U)|=|S| \big)\\
  & =\mathbb{P}\Big(\tau(i)\in f_1(U) \,\Big|\, \big\{ |f_1(U)|=|S| \big\} \cap \big\{ S\backslash\{i\} \subset f_1(U) \big\} \Big) \\
  &\qquad\qquad \cdot\mathbb{P}\Big(S\backslash\{i\} \subset f_1(U) \,\Big|\, |f_1(U)|=|S| \Big) \mathbb{P}\big( |f_1(U)|=|S| \big) \\
  & = \mathbb{P}\big(f_1(U)=\tau(S)\big),
\end{align*}
where in the second to last equality we used Proposition~\ref{prop:symSampling}.

Finally, since any permutation $\sigma$ can be obtained as a sequence of transpositions, applying the previous argument iteratively yields
\[ \mathbb{P}\big(f_1(U)=S\big) = \mathbb{P}\big(f_1(U)=\sigma(S)\big), \]
for every permutation $\sigma$.
\end{proof}

\begin{remark}
   Although all sets of servers of the same size have the same probability of being sampled, the memory state and the incoming job size can influence the number of sampled servers.
\end{remark}

Similarly, if the dispatcher has $o(\log(n))$ bits of memory, then the distribution of the destination of the incoming job is uniform (possibly zero) outside the set of sampled servers.

\begin{proposition}\label{prop:symDispatching}
Let $V$ be a uniform random variable over $[0,1]$. For all $n$ large enough, for every memory state $m\in\mathcal{M}_n$, every set of indices $S\in\mathcal{P}(\{1,\dots,n\})$ with $|S|\in o(n)$, every queue vector state ${\bf q}\in\mathcal{Q}^n$, every rate vector ${\bf \mu}\in\mathbb{R}_+^n$, and  every job size $w \in\mathbb{R}_+$, we have
\begin{align*}
 &\mathbb{P}\Big(f_2\big(m,w,\{({\bf q}_i,{\bf \mu}_i,i):i\in S\},V\big)=j\Big)
\end{align*}
is the same for all $j\notin S$.
\end{proposition}
\begin{proof}
This is a special case of Proposition 5.2 in \cite{negativeResult}.
\end{proof}

\subsection{High arrival rate to slow servers}\label{subsec:highArrivalRate}
In this subsection we will leverage the results of the previous subsection to show that the total workload in the system diverges in time.\\

For every $t\geq 0$, let $\mathcal{W}^n(t)$ be the total remaining workload in the system at time $t$.

\begin{lemma}\label{lem:explodingWorkload}
Fix some $\lambda>0$, and suppose that the service rate of $\lfloor n/2 \rfloor$ servers is equal to some  $\epsilon_n>0$. Then, there exists a positive sequence $\big\{b_n(\lambda)\big\}_{n\geq 1}$, which is completely determined by $\lambda$ (i.e., independent of $\epsilon_n$) such that $b_n\in \Theta\big(e^{-\alpha_n/n}\big)$, and
\[ \liminf_{t\to\infty} \frac{\mathcal{W}^n(t)}{t} \geq \big[b_n(\lambda)-\epsilon_n \big]n, \qquad a.s., \]
for all $n$ large enough.
\end{lemma}

\begin{proof}
Let $\overline{A}_n(t)$ be the counting process of arrivals with a job size of at least $1/2$, and let us define
\[ p_{1/2}\triangleq \mathbb{P}\left( W_1 \geq \frac{1}{2} \right). \]
Since the arrivals are modeled as a renewal process of rate $\lambda n$, and the job sizes $\{W_k\}_{k=1}^\infty$ are i.i.d. with unit mean, it follows that $\overline{A}_n(t)$ is a renewal counting process with rate $\lambda n p_{1/2}>0$. On the other hand, since the average message rate (cf. Equation~\ref{eq:messageRate2}) is upper bounded by $\alpha_n$ almost surely, we have
\begin{equation*}
  \limsup_{t\to\infty} \frac{1}{t} \sum\limits_{k=1}^{\overline{A}_n(t)} 2|S_k| \leq \alpha_n, \qquad a.s.
\end{equation*}
Combining this with the fact that
\begin{align*}
  \limsup_{t\to\infty} \frac{1}{t} \sum\limits_{k=1}^{\overline{A}_n(t)} 2 \left( \frac{\alpha_n}{\lambda n p_{1/2}} \right) \mathds{1}_{\big\{|S_k|> \frac{\alpha_n}{\lambda n p_{1/2}}\big\}} &\leq \limsup_{t\to\infty} \frac{1}{t} \sum\limits_{k=1}^{\overline{A}_n(t)} 2 |S_k|,
\end{align*}
we obtain
\[ \limsup_{t\to\infty} \frac{1}{t} \sum\limits_{k=1}^{\overline{A}_n(t)} \mathds{1}_{\big\{|{S}_k|> \frac{\alpha_n}{\lambda n p_{1/2}}\big\}} \leq \frac{\lambda n p_{1/2}}{2}. \]
This in turn implies that
\begin{align}
  \liminf_{t\to\infty} \frac{1}{t} \sum\limits_{k=1}^{\overline{A}_n(t)} \mathds{1}_{\big\{|{S}_k|\leq \frac{\alpha_n}{\lambda n p_{1/2}}\big\}} &= \liminf_{t\to\infty} \frac{1}{t} \sum\limits_{k=1}^{\overline{A}_n(t)} \left( 1- \mathds{1}_{\big\{|{S}_k|> \frac{\alpha_n}{\lambda n p_{1/2}}\big\}} \right) \nonumber \\
  &\geq \liminf_{t\to\infty} \frac{\overline{A}_n(t)}{t} + \liminf_{t\to\infty} \frac{1}{t} \sum\limits_{k=1}^{\overline{A}_n(t)} - \mathds{1}_{\big\{|{S}_k|> \frac{\alpha_n}{\lambda n p_{1/2}}\big\}} \nonumber \\
  &=\lambda n p_{1/2} - \limsup_{t\to\infty} \frac{1}{t} \sum\limits_{k=1}^{\overline{A}_n(t)} \mathds{1}_{\big\{|{S}_k|> \frac{\alpha_n}{\lambda n p_{1/2}}\big\}} \nonumber \\
  &\geq \frac{\lambda n  p_{1/2}}{2}, \qquad a.s. \label{eq:smallSampling}
\end{align}

Let $N_{\epsilon_n}\subset\{1,\dots,n\}$ be the set of servers with service rate $\epsilon_n$, which was assumed to have cardinality $\lfloor n/2 \rfloor$. Let $s$ be a nonnegative integer upper bounded by $s^*_n \triangleq \alpha_n/\lambda n p_{1/2}$. Since $s^*_n \in o(n)$, Corollary~\ref{cor:symSampling} applies, and we obtain
\begin{align*}
\mathbb{P}\big({S}_k\subset N_{\epsilon_n} \,\big|\, |{S}_k| = s \big)
& =  \frac{{ \lfloor n/2 \rfloor \choose s }}{{n \choose s }} \\
& = \frac{ \lfloor n/2 \rfloor \big( \lfloor n/2 \rfloor - 1 \big) \cdots \big( \left\lfloor n/2 \right\rfloor - s + 1 \big)}{n\big(n-1\big)\cdots \big(n-s +1\big)} \\
& \geq \left(\frac{1}{3}\right)^s,
\end{align*}
for all $n$ large enough, where in the last inequality we used that $s^*_n\in o\big(n\big)$. Since this is true for all $s$ in the given range, we obtain
\[ \mathbb{P}\big({S}_k\subset N_{\epsilon_n} \,\big|\, |{S}_k|\leq s^*_n \big) \geq \left(\frac{1}{3}\right)^{s^*_n}, \]
for all $k\geq 1$, and for all $n$ large enough. Combining this with Equation~\eqref{eq:smallSampling}, we obtain
\begin{equation}
  \liminf_{t\to\infty} \frac{1}{t} \sum\limits_{k=1}^{\overline{A}_n(t)} \mathds{1}_{\big\{|{S}_k|\leq s^*_n, \,\, S_k\subset N_{\epsilon_n}\big\}} \geq \frac{\lambda n  p_{1/2}}{2}\left(\frac{1}{3}\right)^{s^*_n}, \label{eq:slowSampling}
\end{equation}
almost surely, for all $n$ large enough.

Let us fix a particular set $S$ that satisfies
$|S| \leq s^*_n$, and $S\subset N_{\epsilon_n}$. For any such set, Proposition~\ref{prop:symDispatching} implies
\begin{align*}
 \mathbb{P}\big(D_k\in N_{\epsilon_n} \, \big| \, S_k=S \big) &= \mathbb{P}\big(D_k\in N_{\epsilon_n} \, \big| \, D_k \in S,\,\, S_k=S \big) \mathbb{P}\big(D_k\in S \, \big| \, S_k=S \big) \\
 &\quad + \mathbb{P}\big(D_k\in N_{\epsilon_n} \, \big| \, D_k \in S^c,\,\, S_k=S \big) \mathbb{P}\big(D_k\in S^c \, \big| \, S_k=S \big) \\
 &= \mathbb{P}\big(D_k\in S \, \big| \, S_k=S \big) + \frac{|N_{\epsilon_n}\cap S^c|}{|S^c|} \mathbb{P}\big(D_k\in S^c \, \big| \, S_k=S \big) \\
 &\geq \frac{|N_{\epsilon_n}\cap S^c|}{|S^c|} \\
 &=\frac{\left\lfloor \frac{n}{2}\right\rfloor - |S|}{n-|S|} \\
 &\geq \frac{\left\lfloor \frac{n}{2}\right\rfloor - s^*_n}{n} \\
 &\geq \frac{1}{3},
\end{align*}
for all $n$ large enough, where in the last inequality we used that $s^*_n\in o\big(n\big)$. Since this is true for every set $S$ with the given properties, we conclude that
\[ \mathbb{P}\big(D_k\in N_{\epsilon_n} \,\big|\, {S}_k\subset N_{\epsilon_n}, \, |{S}_k|\leq s^*_n \big) \geq \frac{1}{3}. \]
for all $k\geq 1$, and for all $n$ large enough. Combining this with Equation \eqref{eq:slowSampling}, we obtain
\begin{align*}
  \liminf_{t\to\infty} \frac{1}{t} \sum\limits_{k=1}^{\overline{A}_n(t)} \mathds{1}_{\big\{D_k\in N_{\epsilon_n}\big\}} &\geq \liminf_{t\to\infty} \frac{1}{t} \sum\limits_{k=1}^{\overline{A}_n(t)} \mathds{1}_{\big\{D_k\in N_{\epsilon_n}, \, |{S}_k|\leq s^*_n, \, {S}_k\subset N_{\epsilon_n}\big\}} \\
   &\geq \frac{\lambda n  p_{1/2}}{6}\left(\frac{1}{3}\right)^{s^*_n}, \qquad a.s.,
\end{align*}
for all $n$ large enough. Note that this is a lower bound on the average rate of arrival of jobs with size at least $1/2$, to the servers with service rate $\epsilon_n$. On the other hand, those servers have a total processing rate of $\epsilon_n \lfloor n/2 \rfloor$ units of workload per unit of time. Then, since the total workload of the system is at least as much as the workload of the servers with rate $\epsilon_n$, we have
\begin{align*}
 \liminf_{t\to\infty} \frac{\mathcal{W}^n(t)}{t} &\geq \liminf_{t\to\infty} \frac{1}{t} \sum\limits_{k=1}^{\overline{A}_n(t)} \frac{1}{2}\mathds{1}_{\big\{D_k\in N_{\epsilon_n}\big\}} -\epsilon_n \left\lfloor \frac{n}{2} \right\rfloor \\
 & \geq \left[\frac{\lambda p_{1/2}}{6}\left(\frac{1}{3}\right)^{s^*_n} - \epsilon_n  \right] n,
\end{align*}
for all $n$ large enough. This establishes the desired result, with $b_n(\lambda)$ equal to the first term in the bracketed expression above.
\end{proof}

Lemma~\ref{lem:explodingWorkload} implies that, for all $n$ large enough, the total workload in the system increases at least linearly with time, as long as $\lfloor n/2 \rfloor$ of the servers have rate $\epsilon_n<b_n(\lambda)$. In particular, this will happen if $\epsilon_n\in O\big(e^{-\alpha_n/n}\big)$.

 Since the above is true for every weakly symmetric policy with $o(\log n)$ bits of memory, and with an average message rate upper bounded by $\alpha_n\in o\big(n^2\big)$ almost surely, it follows that, for all $n$ large enough, the stability region of any such policy is contained in a proper subset
 $\Gamma_n(\lambda,\alpha_n)$ of $\Sigma_n$ which excludes service rate vectors for which $\lfloor n/2 \rfloor$ of the servers have rate $\epsilon_n<b_n(\lambda)$.

\bibliographystyle{imsart-number}
\bibliography{references}
\end{document}